\documentclass[conference]{IEEEtran}
\usepackage{blindtext, graphicx}
\usepackage{tabularx}
\usepackage{pgf,tikz,pgfplots}
\usetikzlibrary[patterns]
\usepackage[ruled,vlined]{algorithm2e}
\usepackage{hyperref}
\pgfplotsset{compat=1.14}
\usepackage{mathrsfs}
\usetikzlibrary{arrows}
\hypersetup{
 colorlinks,
 linkcolor={blue!100!black},
 citecolor={blue!100!black},
 urlcolor={blue!80!black}
}

\usepackage{amsmath,amssymb,amsfonts,amsthm,graphicx,bm,bbm,dsfont}
\newtheorem{theorem}{Theorem}
\newtheorem{remark}{Remark}

\newtheorem{proposition}{Proposition}

\newtheorem{lemma}{Lemma}

\DeclareSymbolFont{bbold}{U}{bbold}{m}{n}
\DeclareSymbolFontAlphabet{\mathbbold}{bbold}
\newcommand{\1}{\mathbbold{1}}

\usepackage{color}

\newcommand{\cB}{\mathcal{B}}

\newcommand{\cR}{\mathcal{R}}

\newcommand{\boldw}{\mathbf{w}}
\newcommand{\boldx}{\mathbf{x}}
\newcommand{\boldy}{\mathbf{y}}

\makeatletter
\def\namedlabel#1#2{\begingroup
	\def\@currentlabel{#2}%
	\label{#1}\endgroup
}
\makeatother
\begin{document}
\title{Trace Reconstruction with Bounded Edit Distance}

\author{\textbf{Jin Sima} \IEEEauthorblockN{ and \textbf{Jehoshua Bruck}}\\
	\IEEEauthorblockA{
	Department of Electrical Engineering,
California Institute of Technology, Pasadena 91125, CA, USA\\}}

\maketitle

\begin{abstract}
The trace reconstruction problem studies the number of noisy samples needed to recover an unknown string $\boldx\in\{0,1\}^n$ with high probability, where the samples are independently obtained by passing $\boldx$ through a random deletion channel with deletion probability $q$. The problem is receiving significant attention recently due to its applications in DNA sequencing and DNA storage. Yet, there is still an exponential gap between upper and lower bounds for the trace reconstruction problem. In this paper we study the trace reconstruction problem when $\boldx$ is confined to an edit distance ball of radius $k$, which is essentially equivalent to distinguishing two strings with edit distance at most $k$. It is shown that $n^{O(k)}$ samples suffice to achieve this task with high probability.
\end{abstract}
\pagenumbering{gobble}



\section{Introduction}\label{section:introduction}
The trace reconstruction problem seeks to recover an unknown string $\boldx\in\{0,1\}^n$, given multiple independent noisy samples or traces of $\boldx$. In this paper, a noisy sample is obtained by passing $\boldx$ through a deletion channel, which randomly and independently deletes each bit of $\boldx$ with probability $q$. We are interested in how many samples are needed to recover $\boldx$ with high probability.

The trace reconstruction problem was introduced in \cite{Batu} and proposed earlier in \cite{Lev} under an adversarial setting. It has been receiving increased attention recently due to its application in DNA sequencing \cite{BPRS} and DNA storage under nanopore sequencing \cite{Microsoft,Yazdi}. Also, there are many significant results on trace reconstruction and its variants and generalizations, such as coding for trace reconstruction \cite{Ryan} and population recovery \cite{Ban}.
For average case trace reconstruction, where the reconstruction error probability is averaged over all choices of $\boldx\in\{0,1\}^n$, 
the state of the art upper and lower bounds on the number of samples are $\exp(O(\log^{\frac{1}{3}}(n)))$ \cite{Holden} and $\Omega(\frac{\log^{\frac{5}{2}}(n)}{(\log\log n)^7})$ \cite{Chasel} respectively.
  
Despite the progress for average cases, the trace reconstruction problem proved to be highly nontrivial in worst cases, where the reconstruction error probability goes to zero for arbitrary choice of $\boldx$. For small deletion probabilities, the work in \cite{Rocco} showed that polynomial number of samples suffice when $q\le n^{-(\frac{1}{3}+\epsilon)}$ for some $\epsilon>0$, improving the result in \cite{Batu} for $q\le n^{-(\frac{1}{2}+\epsilon)}$ and some $\epsilon >0$. When the deletion probability becomes constant,
there is still an exponential gap between the upper and lower bounds on the number of samples needed. The first achievable sample size for constant deletion probability $q$ is $\exp(\tilde{O}(n^\frac{1}{2}))$ \cite{Mitzenmacher}, which was improved to $\exp(\Theta(n^\frac{1}{3}))$ in independent and simultaneous works \cite{De} and \cite{Peres}. Both \cite{De} and \cite{Peres} studied mean-based algorithms, which use single bit statistics in traces, for reconstruction. They showed that $\exp(O(n^\frac{1}{3}))$ is the best sample size achieved by mean-based algorithms. A novel approach in \cite{De} and \cite{Peres} is to relate single-bit statistics to complex polynomial analysis, and borrow results from \cite{BE97} on complex analysis. This approach was further developed in \cite{Chaseu}, where multi-bit statistics were considered.  
The current best upper bound on the sample size is $\exp(\tilde{O}(n^\frac{1}{5}))$ \cite{Chaseu}, while the best lower bound $\Omega(\frac{n^{\frac{3}{2}}}{\log^7n})$ \cite{Chasel} is orders of magnitude away from the upper bound. 

While the general trace reconstruction problem is hard to solve, in this paper, we focus on a variant of the trace reconstruction problem with an edit distance constraint. Specifically, the goal is to recover the string $\boldx$ by using its noisy samples and additional information of a given string $\boldy$, which is known to be within a bounded distance from $\boldx$. The edit distance between two strings is commonly defined as the minimum number of deletions, insertions, or substitutions that transform one string into another. In this paper, we consider only deletion/insertion for convenience, as a substitution is an insertion followed by a deletion. We say that a string $\boldx$ is within edit distance $k$ to a string $\boldy$, denoted as $\boldx\in\cB_k(\boldy)$, if $\boldx$ can be obtained from $\boldy$ after at most $k$ deletions and $k$ insertions. Note that the general trace reconstruction problem considers cases where $k=n$. 

The setting considered in this paper arises in many practical scenarios in genome sequencing, where one needs to recover an individual genome sequence of a species, given a reference genome sequence that represents the species \cite{Wiki}. Normally, the genome sequences of a species share some similarity and most of them can be considered to be within a bounded edit distance from the reference genome. 
One example is the Human Genome Project, where a human reference genome is provided to study the difference between individual genomes. 
Complementary to the problem we consider, the work in \cite{Racz} studied approximate trace reconstruction, which aims to find an estimate within a given edit distance to the true string. Note that such an estimate, together with an algorithm to distinguish two strings within edit distance $k$, establishes a solution to the general trace reconstruction problem.

As indicated in \cite{De,Sudan,Mitzenmacher,Mazumdar,Peres}, the problem of worst case trace reconstruction is essentially equivalent to a hypothesis testing problem of distinguishing any two strings using noisy samples. More specifically, the sample complexity needed for trace reconstruction is at most $poly(n)$ times the sample complexity needed to distinguish arbitrary two strings. The same equivalence holds in our setting as well, where a reference string $\boldy$ is known and close to $\boldx$ in edit distance. Hence, for convenience, we consider the problem in the form of distinguishing any two strings $\boldx\in\{0,1\}^n$ and $\boldy\in\{0,1\}^n$ when $\boldx$ is within edit distance $k$ to $\boldy$. One special case of the problem 
is to distinguish two strings within Hamming distance $k$, which was addressed in \cite{Mazumdar} and $n^{O(k)}$ sample complexity was achieved. Recently, an independent work \cite{Sudan} studied the limitations of mean-based algorithms (see \cite{De} and \cite{Peres}) in distinguishing two strings with bounded edit distance. It was shown that mean-based algorithms need at least $n^{O(\log n)}$ traces to distinguish two strings with edit distance of even $4$. The paper \cite{Sudan} also showed that $n^{O(k^2)}$ suffices to distinguish two strings $\boldx\in\{0,1\}^n$ and $\boldy\in\{0,1\}^n$ with special block structures, if $\boldx\in\cB_k(\boldy)$. Yet, as pointed out in \cite{Racz}, it is an open problem whether $n^{O(k)}$ samples suffice to recover a string that is within edit distance $k$ to a known string.

The main contribution of this paper is an affirmative answer to this question. We show that distinguishing two sequences within edit distance $k$ needs at most $n^{O(k)}$ samples. The result is stated in the following.
\begin{theorem}\label{theorem:main}
Let $\boldx\in\{0,1\}^n$ and $\boldy\in\{0,1\}^n$ be two strings satisfying $\boldx\in\cB_k(\boldy)$. Then strings $\boldx$ and $\boldy$ can be distinguished with high probability, given $n^{O(k)}$ independent noisy samples, each obtained by passing $\boldx$ through a deletion channel with deletion probability $q<1$. 	
\end{theorem}
\begin{remark}
	Theorem \ref{theorem:main} holds for any string $\boldy$ that can be obtained from $\boldx$ after at most $k$ deletions or insertions. The length of $\boldy$ is not necessarily $n$. Yet by definition of the trace reconstruction problem, we focus on length $n$ strings $\boldx$ and $\boldy$.
\end{remark}
The approach we take follows a similar method to that in \cite{Chaseu,De,Sudan,Peres}, in the sense that we derive bounds on multi-bit statistics through complex analysis of a special class of polynomials. Yet, the complex analysis in this paper differs from those in \cite{Chaseu,De,Sudan,Peres} in the following two ways. Firstly, we make use of the fact that the polynomial is related to a number theoretic problem called the Prouhet-Tarry-Escott problem \cite{Borwein}, which is also noted in \cite{Sudan}. This allows us to link the problem to our previous result on deletion codes \cite{SB20}, where we showed that two constrained strings can be distinguished using weighted sums of powers, which is similar in form to the Prouhet-Tarry-Escott problem. Secondly, to find the maximum value of the polynomial, we let the complex variable take values on a small circle around the point $1$, 
while the work in \cite{Chaseu,De,Sudan,Peres} analyze the complex polynomial on a unit circle. By doing this, we are able to improve the $n^{O(k^2)}$ bound in \cite{Sudan} to $n^{O(k)}$. 

The rest of the paper is organized as follows. In Section \ref{section:preliminaries} we provide an introduction to the techniques and the lemmas needed to prove Theorem \ref{theorem:main}. In Section \ref{section:proofoftheorem1}, the proof of Theorem \ref{theorem:main} is given. Section \ref{section:proofoflemma6} presents the proof of a critical lemma on complex analysis. Section \ref{section:conclusion} concludes the paper.  
\section{Proof Techniques and Lemmas}\label{section:preliminaries}
In this section we present a brief introduction to the techniques and key lemmas needed in proving Theorem 1. For strings $\boldx\in\{0,1\}^n$ and $\boldy\in\{0,1\}^n$, let $\tilde{X}=(\tilde{X}_1,\ldots,\tilde{X}_n)$ and $\tilde{Y}=(\tilde{Y}_1,\ldots,\tilde{Y}_n)$ denote the sample obtained by passing $\boldx$ and $\boldy$ through the deletion channel respectively. We have $\tilde{X}_i=\emptyset$ or $\tilde{Y}_j=\emptyset$ if $i$ or $j$ is larger than the length of $\tilde{X}$ or $\tilde{Y}$, respectively. Note that $\tilde{X}$ and $\tilde{Y}$ are sequences of random variables that describe the probability distributions of the samples.

The techniques we use were originated in \cite{De,Peres}, which  
presented the following identity
\begin{align}\label{equation:singlebit}
	\mathbb{E}_{\tilde{X}}\Big[\sum^n_{i=1}\tilde{X}_i(\frac{z-q}{1-q})^i\Big]&=(1-q)\sum^n_{i=1}x_iz^i\nonumber\\
	&\triangleq f^s_{\boldx}(z),
\end{align}
for a sequence $\boldx$ and a complex number $z$. The identity \eqref{equation:singlebit} links the analysis of single bit statistics $\{E_{\tilde{X}}[\tilde{X}_i]\}^n_{i=1}$ to that of complex polynomials. As a result, a lower bound on the maximal difference between single bit statistics $\max_{1\le i\le n}|E_{\tilde{X}}[\tilde{X}_i]-E_{\tilde{Y}}[\tilde{Y}_i]|$ can be obtained through analyzing the maximal value of the polynomial $f^{s}_{\boldx}(z)-f^{s}_{\boldy}(z)$ on a unit disk, a problem referred to as Littlewood type problems and studied in \cite{BE97,BE}. 
Generalizing the approach in \cite{De,Peres}, the papers \cite{Chaseu} and \cite{Chen} presented multi-bit statistics counterparts of \eqref{equation:singlebit}. In this paper, we consider the version from \cite{Chaseu}, stated in the following lemma. 
\begin{lemma}\label{lemma:multibitgeneral} \cite{Chaseu}
	For integer $\ell\ge 1$, complex numbers $z_1,\ldots,z_\ell$, and sequences $\boldx\in \{0,1\}^n$ and $\boldw\in \{0,1\}^\ell$, we have
	\begin{align}\label{equation:multibitgeneral}
		&\mathbb{E}_{\tilde{X}}\Big{[}(1-q)^{-\ell}\sum_{1\le i_1<\ldots<i_{\ell}\le n}\1_{\tilde{X}_{i_j}=w_j,\forall j\in[\ell]}\nonumber\\
		&~~~~~~~~~~~~~~~~~~~~~~~\cdot(\frac{z_1-q}{1-q})^{i_1}\prod^{\ell}_{j=2}(\frac{z_j-q}{1-q})^{i_j-i_{j-1}-1}\Big{]}\nonumber\\
		=&\sum_{1\le j_1<\ldots<j_{\ell}\le n}\1_{x_{j_h}=w_h,\forall h\in[\ell]}z^{j_1}_1\prod^\ell_{h=1}z^{j_h-j_{h-1}-1}_j\nonumber\\
		\triangleq& f_{\boldx,\boldw}(z_1,\ldots,z_\ell),		
	\end{align}
	where $[\ell]=\{1,\ldots,\ell\}$ and $i:i+\ell-1=\{i,\ldots,i+\ell-1\}$.
	For any statement $E$, the variable $\1_E=1$ iff $E$ holds true.  
\end{lemma}
     By taking $z_2=\ldots=z_\ell=0$ in \eqref{equation:multibitgeneral}, we obtain
	\begin{align}\label{equation:multibit}
		f_{\boldx,\boldw}(z,0,\ldots,0)=\sum^{n-\ell+1}_{i=1}\1_{\boldx_{i:i+\ell-1}=\boldw}z^i.		
	\end{align}    
Similar to the arguments in \cite{Chaseu}, we prove Theorem \ref{theorem:main} by analyzing the polynomial $f_{\boldx,\boldw}(z,0,\ldots,0)-f_{\boldy,\boldw}(z,0,\ldots,0)$ associated with the multi-bit statistics in \eqref{equation:multibit}. Note that the polynomial $f_{\boldx,\boldw}(z,0,\ldots,0)-f_{\boldy,\boldw}(z,0,\ldots,0)$ is single variate. The way in which the polynomial is analyzed in this paper deviates from that in \cite{Chaseu}. While the paper \cite{Chaseu} taylored the complex analysis arguments in \cite{BE} to obtain improved bounds, in this paper, we exploit number theoretic properties of two strings $\boldx$ and $\boldy$ within edit distance $k$. 

In our previous paper \cite{SB20}, we showed implicitly that the weighted sums of powers $\sum^n_{i=1}i^jx_i$, $j\in\{0,\ldots,O(k)\}$  can be used to distinguish two constrained strings $\boldx$ and $\boldy$ within edit distance $k$. The following lemma makes this statement explicit.   
Let $\mathcal{R}_{n,k}$ denote the set of length $n$ strings such that any two $1$ entries in each string are separated by a $0$ run of length at least $k-1$.
\begin{lemma}\label{lemma:parity} 
	For distinct strings~$\boldx,\boldy\in \mathcal{R}_{n,6k}$, if~$\boldx\in\cB_{6k}(\boldy)$, then there exists an integer $m\in [12k+1]$ such that $\sum^n_{i=1}i^mx_i\ne \sum^n_{i=1}i^my_i$.
\end{lemma}
\begin{proof}
	Suppose on the contrary, we have that $\sum^n_{i=1}i^mx_i= \sum^n_{i=1}i^my_i$ for all $m\in[12k+1]$. Then, we have that
\begin{align}\label{equation:sumofpowers}
	\sum^n_{i=1}\Big(\sum^i_{j=1}j^{m'}\Big)x_i= \sum^n_{i=1}\Big(\sum^i_{j=1}j^{m'}\Big)y_i
\end{align} 
for all $m'\in\{0,\ldots,12k\}$. This is because $\sum^i_{j=1}j^{m'}$ is a weighted sum of $i^1,\ldots,i^{m'+1}$ for any $m'\in\{0,\ldots,12k+1\}$ (Faulhaber's formula). Next, we borrow a result from \cite{SB20}.
	\begin{proposition}\label{proposition:parity}\cite{SB20}
		For sequences~$\boldx,\boldy\in \mathcal{R}_{n,3k}$, if~$\boldy\in \mathcal{B}_{3k}(\boldx)$ and~$ \sum^n_{i=1}(\sum^i_{j=1}j^m)x_i= \sum^n_{i=1}(\sum^i_{j=1}j^m)y_i$ for $m\in \{0,\ldots,6k\}$, then~$\boldx=\boldy$.
	\end{proposition}
	Note that $\cB_{5k}(\boldx)\subseteq \cB_{6k}(\boldx)$. Since \eqref{equation:sumofpowers} holds, we apply Proposition \ref{proposition:parity} with $k=2k$ and conclude that $\boldx=\boldy$, which contradicts the fact that $\boldx$ and $\boldy$ are distinct. 	
\end{proof}
Interestingly, 
the following result from \cite{Borwein} connects the sums of powers of two sets of integers that appear in Lemma \ref{lemma:parity} to the number of roots of a polynomial at $1$. It allows us to combine the number theoretic result with further complex analysis, which will be given in Lemma \ref{lemma:complexanalysis}. The lemma can be proved by checking the $i$-th, $i\in[m]$, derivative of the polynomial $\sum^s_{i=1}z^{\alpha_i}- \sum^t_{i=1}z^{\beta_i}$ at point $z=1$.
\begin{lemma}\label{proposition:divide}\cite{Borwein}
	Let $\{\alpha_1,\ldots,\alpha_s\}$ and $\{\beta_1,\ldots,\beta_s\}$ be two sets of integers. 
	The following are equivalent:
	\begin{enumerate}
		\item[](a) $\sum^s_{i=1}\alpha^j_i= \sum^s_{i=1}\beta^j_i$ for $j\in [m-1]$.
		\item[](b) $(z-1)^m \textup{ divides } \sum^s_{i=1}z^{\alpha_i}- \sum^s_{i=1}z^{\beta_i}$. 
	\end{enumerate}
\end{lemma}
\begin{remark}
The problem of finding two sets of integers $\{\alpha_1,\ldots,\alpha_s\}$ and $\{\beta_1,\ldots,\beta_s\}$ satisfying the statement (a) is called the Prouhet-Tarry-Escott problem \cite{Borwein}. This connection between the Prouhet-Tarry-Escott problem and the analysis of polynomials was also used in \cite{Sudan} and implicitly in \cite{KR97}. 
\end{remark} 
Lemma \ref{lemma:parity} requires that the strings $\boldx$ and $\boldy$ are within $\cR(n,6k)$, which does not hold in general. Following the same trick as in \cite{Chaseu} and \cite{SB20}, we define an indicator vector as follows. For any sequences $\boldx\in\{0,1\}^n$ and $\boldw\in\{0,1\}^\ell$, define the length $n$ vector
\begin{align*}
	\1_{\boldw}(\boldx)_i &\triangleq \begin{cases}
		1, &\text{if~$\boldx_{i:i+\ell-1}=\boldw$,}\\
		0,&\text{else.}\\
	\end{cases}
\end{align*}
for $i\in[n]$. Note that $\1_{\boldw}(\boldx)_i=0$ for $i\in\{n-\ell+2,\ldots,n\}$. 
It can be seen that the polynomial $f_{\boldx,\boldw}(z,0,\ldots,0)$ related to multi-bit statistics is exactly the polynomial $f^s_{\1_{\boldw}(\boldx)}(z)$ related to single-bit statistics.
To apply Lemma \ref{lemma:parity}, we need to find a $\boldw$ such that $\1_{\boldw}(\boldx)\in\cR(n,6k)$. The same as what the paper \cite{Chaseu} did, we find such a $\boldw$ by using the following lemma from \cite{Robson}. A string $\boldw\in\{0,1\}^\ell$ is said to have period $a$, if and only if $w_i=w_{i+a}$ for $i\in [\ell-a]$.  Moreover, a string $\boldw\in\{0,1\}^\ell$ is said to be non-periodic, iff $\boldw$ does not have period $a$ for $a\in[\lceil \frac{\ell}{2}\rceil-1]$. 
\begin{lemma}\label{lemma:period}
	For any sequences~$\boldw\in\{0,1\}^{2p-1}$, either~$(\boldw,0)$ or~$(\boldw,1)$ is non-periodic, where $(\boldw,0)$ and $(\boldw,1)$ is the string obtained by appending $0$ and $1$ to $\boldw$, respectively. 
\end{lemma}
Lemma \ref{lemma:period} can be proved by definition of period. 
The claim that $\1_{\boldw}(\boldx)\in\cR(n,p)$ follows from Lemma \ref{lemma:period} and will be proved in Lemma \ref{lemma:distancek}. In addition, the edit distance between $\1_{\boldw}(\boldx)$ and $\1_{\boldw}(\boldy)$ needs to be bounded to apply Lemma \ref{lemma:parity}. This is proved in the following lemma.
\begin{lemma}\label{lemma:distancek}
	Let~$\boldw\in\{0,1\}^{2p}$ be a non-periodic string. 
	For two strings~$\boldx$ and~$\boldy\in\cB_{k}(\boldx)$, we have that
	\begin{enumerate}
		\item[(a)] $\1_{\boldw}(\boldx)\in \mathcal{R}_{n,p}$.
		\item[(b)] $\1_{\boldw}(\boldy)\in \mathcal{R}_{n,p}$.
		\item[(c)] $\1_{\boldw}(\boldx)\in\cB_{5k}(\1_{\boldw}(\boldy))$.
	\end{enumerate}
\end{lemma}
\begin{proof}
The statements (a) and (b) follow from the definition of vectors $\1_{\boldw}(\boldx)$ and $\1_{\boldw}(\boldy)$ and the fact that $\boldw$ is non-periodic. Suppose there are two $1$ entries $\1_{\boldw}(\boldx)_i$ and $\1_{\boldw}(\boldx)_{i+a}$ in $\1_{\boldw}(\boldx)$ that are separated by less than $p-1$ $0$'s, i.e., $a\le p-1$. Then by definition of $\1_{\boldw}(\boldx)$, we have that $\boldx_{i:i+2p-1}=\boldw$ and that $\boldx_{i+a:i+a+2p-1}=\boldw$. This implies that $w_{j}=x_{i+a+j-1}=w_{j+a}$ for $j\in[2p -a]$. Hence, the string $\boldw$ has period $a\le p-1$, contradicting to the fact that $\boldw$
is non-periodic. Hence, we have that $\1_{\boldw}(\boldx)\in \mathcal{R}_{n,p}$, and similarly that $\1_{\boldw}(\boldy)\in \mathcal{R}_{n,p}$

We now prove statement (c).
To this end, we first show that a deletion in $\boldx$ results in at most three deletions and two insertions in $\1_{\boldw}(\boldx)$. Since $\boldw$ has length $2p$ and $\1_{\boldw}(\boldx)\in \mathcal{R}_{n,p}$ as shown in (a), a deletion in $\boldx$ results in at most two deletions and two insertions of $1$ entries in $\1_{\boldw}(\boldx)$, respectively. Otherwise, suppose that a deletion in $\boldx$ deletes three $1$ entries $\1_{\boldw}(\boldx)_{i_1}$, $\1_{\boldw}(\boldx)_{i_2}$, and $\1_{\boldw}(\boldx)_{i_3}$ in $\1_{\boldw}(\boldx)$, then we have that $i_3-i_1\ge 2p$ because (a) holds. This is impossible since $\boldw\in\{0,1\}^{2p}$ and the deletion in $\boldx$ can not affect the two occurrences $\boldx_{i_1:i_1+2p-1}$ and $\boldx_{i_3:i_3+2p-1}$ of $\boldw$ in $\boldx$ simultaneously. Hence a deletion causes at most two deletions of $1$ entries in $\1_{\boldw}(\boldx)$ and similarly, the same holds for insertions.
 
Moreover, at most one $0$ entry is deleted in $\1_{\boldw}(\boldx)$ because of the deletion in $\boldx$. Hence, a  deletion in $\boldx$ causes at most three deletions and two insertions in total in $\1_{\boldw}(\boldx)$, and $k$ deletions in $\boldx$ results in at most $3k$ deletions and $2k$ insertions in $\1_{\boldw}(\boldx)$. 
The same holds for $\boldy$ and $\1_{\boldw}(\boldy)$.  	

Since $\boldx\in \cB_{k}(\boldy)$, we conclude that $\1_{\boldw}(\boldy)$ can be obtained from $\1_{\boldw}(\boldx)$ by at most $5k$ deletions and $5k$ insertions, and hence, $\1_{\boldw}(\boldx)\in \cB_{5k}(\1_{\boldw}(\boldy))$.
\end{proof}
With Lemma \ref{lemma:parity} and Lemma \ref{lemma:distancek} established, we present a lower bound on the maximal value of polynomial
$f_{\boldx,\boldw}(z,0,\ldots,0)-f_{\boldy,\boldw}(z,0,\ldots,0)$ for $z$ close to $1$. Note that it is important that $z$ is located near the point $1$ on the complex plane because of the scaling factor $(\frac{z-q}{1-q})^i$ in the multi-bit statistics in Eq. \eqref{equation:multibit}. To meet this requirement on $z$, existing works \cite{Chaseu,De,Sudan,Peres} restrict $z$ to lie on short subarcs of a unit circle around $1$, a case also considered in \cite{BE97} in the context of complex analysis. In this paper, we choose $z$ from a small circle around $1$. It turns out that this choice of $z$ achieves a lower bound $\frac{1}{n^{O(k)}}$ on $f_{\boldx,\boldw}(z,0,\ldots,0)-f_{\boldy,\boldw}(z,0,\ldots,0)$, which improves the bound $\frac{1}{n^{O(k^2)}}$ established in \cite{Sudan}. The details will be given in the following lemma, which is a critical result in this paper. Its proof will be given in Section \ref{section:proofoflemma6}.
\begin{lemma}\label{lemma:complexanalysis}
	For integer $\ell\ge 1$ and strings~$\boldx,\boldy\in\{0,1\}^n$ and $\boldw\in\{0,1\}^\ell$, if $\sum^{n}_{i=1}\1_{\boldw}(\boldx)_ii^m\ne \sum^{n}_{i=1}\1_{\boldw}(\boldy)_ii^m$ for some non-negative integer $m$, then there exists a complex number $z$, such that $|\frac{z-q}{1-q}|^n\le 2$ and
	\begin{align}\label{equation:complexanalysis}
		 \sum^{n}_{i=1}\1_{\boldw}(\boldx)_iz^i
		-\sum^{n}_{i=1}\1_{\boldw}(\boldy)_iz^i\ge \frac{1}{n^{2m}(2m+2)}.
	\end{align}
for sufficiently large $n$.
\end{lemma}
Finally, we use the lower bound in Lemma \ref{lemma:complexanalysis} for single variate polynomial $f_{\boldx,\boldw}(z,0,\ldots,0)-f_{\boldy,\boldw}(z,0,\ldots,0)$ to obtain a lower bound for the multi-variate polynomial $f_{\boldx,\boldw}(z_1,\ldots,z_\ell)-f_{\boldy,\boldw}(z_1,\ldots,z_\ell)$, where $z_1,\ldots,z_\ell$ are close to $1$.   
This lower bound guarantees a gap between the multi-bit statistics of $\tilde{X}$ and $\tilde{Y}$, which makes $\boldx$ and $\boldy$ distinguishable by Hoeffding's inequality (See Section \ref{section:proofoftheorem1}). The proof follows similar steps to the ones in \cite{Chaseu}.
\begin{lemma}\label{lemma:complexanalysisgeneral}
	For integer $\ell\ge 1$ and strings~$\boldx,\boldy\in\{0,1\}^n$ and $\boldw\in\{0,1\}^\ell$, if $\sum^{n}_{i=1}\1_{\boldw}(\boldx)_ii^m\ne \sum^{n}_{i=1}\1_{\boldw}(\boldy)_ii^m$ for some non-negative integer $m$, then there exist complex numbers $z_1,\ldots,z_\ell$, such that $|\frac{z_i-q}{1-q}|^n\le 2$ for $j\in[\ell]$ and
	\begin{align}\label{equation:complexanalysisgeneral}
		f_{\boldx,\boldw}(z_1,\ldots,z_\ell)-f_{\boldy,\boldw}(z_1,\ldots,z_\ell)\ge \frac{1}{n^{O(m)}}.
	\end{align}
for sufficiently large $n$.
\end{lemma}
\begin{proof}
	According to Lemma \ref{lemma:complexanalysis}, there exists a complex number $z^*$ satisfying $|\frac{z^*-q}{1-q}|^n\le 2$ and \eqref{equation:complexanalysis}. Let $z_1=z^*$ and $z_2=\ldots=z_\ell=z$. Then the polynomial $f(z^*,z)\triangleq f_{\boldx,\boldw}(z^*,z,\ldots,z)-f_{\boldy,\boldw}(z^*,z,\ldots,z)$ is a function of $z$. 
	By \eqref{equation:multibit} and \eqref{equation:complexanalysis} we have that $f(z^*,0)\ge \frac{1}{n^{2m}(2m+2)}$. The following result from \cite{BE} relates $f(z^*,0)$ to the maximal value of $f(z^*,z)$ for $z$ close to $1$.
	\begin{proposition}\label{proposition:5.1}\cite{BE}
		Let $f(z)$ be an analytic function satisfying $f(z)\le \frac{1}{1-|z|}$ for $|z|<1$. There are positive real constants $c_1$ and $c_2$ such that
		\begin{align*}
			|f(0)|^{\frac{c_1}{a}}\le \exp(\frac{c_2}{a})\max_{z\in [1-a,1]}|f(z)|
		\end{align*}
	for real number $a\in (0,1]$ 
	\end{proposition} 
    According to Proposition \ref{proposition:5.1}, we have that
    \begin{align}\label{equation:max}
    &\max_{z\in [\max\{2q-1,0\},1]}|f(z^*,z)|\\
    \ge&	\exp(-\frac{c_2}{1-\max\{2q-1,0\}})|f(z^*,0)|^{\frac{c_1}{1-\max\{2q-1,0\}}}\\
    \ge&O(\frac{1}{n^{O(m)}})    
\end{align}
    Let $z_1=z^*$ and $z_2=\ldots=z_\ell$ be the number $z$ maximizing the term $|f(z^*,z)|$ in \eqref{equation:max}. Then by Lemma \ref{lemma:complexanalysis} we have that $|\frac{z_1-q}{1-q}|^n\le 2$ for sufficiently large $n$ and $|\frac{z_i-q}{1-q}|^n\le 1$ for $i\in\{2,\ldots,\ell\}$. Hence, the proof is done.
\end{proof}
\section{Proof of Theorem 1}\label{section:proofoftheorem1}
In this section we prove Theorem 1 based on the results from Lemma \ref{lemma:multibitgeneral} to Lemma \ref{lemma:distancek} and Lemma \ref{lemma:complexanalysisgeneral}.  
Let $t_0$ be the smallest index such that $x_i\ne y_i$. If $t_0 <12k$, we have the following result from \cite{Zhai}, which was also used in \cite{Chaseu}. 
\begin{proposition}\label{proposition:elllessthank}
	For sequences~$\boldx,\boldy\in\{0,1\}^n$, let $t_0$ be the smallest index such that $x_{t_0}\ne y_{t_0}$, i.e., $x_i=y_i$ for $i\in [t_0-1]$. Then, with high probability $\boldx$ and $\boldy$ can be distinguished using $\exp(O(t_0^{\frac{1}{3}}))$ samples. 
\end{proposition}
According to Proposition \ref{proposition:elllessthank}, 
sequences $\boldx$ and $\boldy$ can be distinguished with high probability using $\exp(O(t_0^{\frac{1}{3}}))<n^{O(k)}$ samples. Hence, it suffices to consider cases when $t_0\ge 12k$.

Let $\boldw'=\boldx_{t_0-12k+1:t_0-1}$. By Lemma \ref{lemma:period}, either $(\boldw',0)$ or $(\boldw',1)$ is non-periodic. Without loss of generality, assume that $\boldw=(\boldw',0)\in\{0,1\}^{12k}$ is non-periodic. Then,
similar to the arguments in \cite{Chaseu,De,Sudan,Mazumdar,Peres}, the core part of the proof is to show that the difference of multi-bit statistics $\mathbb{E}_{\tilde{X}}[\1_{\tilde{X}_{i_j}=w_j,\forall j\in[12k]}]$ and $\mathbb{E}_{\tilde{Y}}[\1_{\tilde{Y}_{i_j}=w_j,\forall j\in[12k]}]$, is at least $\frac{1}{n^{O(k)}}$ for some integers $1\le i_1<\ldots<i_{12k}\le n$, i.e.,
\begin{align}\label{equation:totalvariance}
	\max_{1\le i_1<\ldots<i_{12k}\le n}\Big|&\mathbb{E}_{\tilde{X}}\big[\1_{\tilde{X}_{i_j}=w_j,\forall j\in[12k]}]\nonumber\\
	-&\mathbb{E}_{\tilde{Y}}[\1_{\tilde{Y}_{i_j}=w_j,\forall j\in[12k]}\big]\Big|\ge \frac{1}{n^{O(k)}}.	
\end{align}
Let 
\begin{align*}
	(i^*_1,\ldots,i^*_{12k})=\text{argmax}_{1\le i_1<\ldots<i_{12k}\le n}&|\mathbb{E}_{\tilde{X}}[\1_{\tilde{X}_{i_j}=w_j,\forall j\in[12k]}]\\
	-&\mathbb{E}_{\tilde{Y}}[\1_{\tilde{Y}_{i_j}=w_j,\forall j\in[12k]}]|,
\end{align*}  
which can be determined once $\boldx$ and $\boldy$ are given. Suppose that $\boldx$ is passed through the deletion channel $N$ times,  generating $N$ independent samples $\{\tilde{T}^t\}^N_{t=1}$.
Then,  	
by using similar Hoeffding's inequality (or the Chernoff bound) arguments as in \cite{Peres}, we can show that with high probability, the empirical 
distribution $\frac{\sum^N_{t=1}\1_{\tilde{T}^t_{i^*_j}=w_j,\forall j\in[12k]}}{N}$ is closer to 
$E\big[\1_{\tilde{X}_{i^*_j}=w_j,\forall j\in[12k]}\big]$ than to $E\big[\1_{\tilde{Y}_{i^*_j}=w_j,\forall j\in[12k]}\big]$, if 
\begin{align*}
N&\ge O\left(\frac{1}{|\mathbb{E}_{\tilde{X}}[\1_{\tilde{X}_{i^*_j}=w_j,\forall j\in[12k]}]-\mathbb{E}_{\tilde{Y}}[\1_{\tilde{Y}_{i^*_j}=w_j,\forall j\in[12k]}]|^2}\right)\\
&=n^{O(k)}.
\end{align*}
Hence $\boldx$ and $\boldy$ can be distinguished using $n^{O(k)}$ samples. Therefore, it suffices to show \eqref{equation:totalvariance} in the rest of the proof.

Since $\boldw$ is non-periodic and $\boldx\in\cB_k(\boldy)$, Lemma \ref{lemma:distancek} implies that $\1_{\boldw}(\boldx),\1_{\boldw}(\boldy)\in\cR(n,6k)$ and that $\1_{\boldw}(\boldx)\in \cB_{5k}(\1_{\boldw}(\boldy))$. In addition,  either $\boldx_{t_0-12k+1:t_0}=\boldw$ or $\boldy_{t_0-12k+1:t_0}=\boldw$ holds by definition of $\boldw$ and $t_0$. Therefore, we have that $\1_{\boldw}(\boldx)_{t_0-12k+1}\ne \1_{\boldw}(\boldy)_{t_0-12k+1}$, and thus that $\1_{\boldw}(\boldx)\ne \1_{\boldw}(\boldy)$. Hence, we apply Lemma \ref{lemma:parity} and obtain an integer $m\in [12k+1]$ such that $\sum^{n}_{i=1}\1_{\boldw}(\boldx)_ii^m\ne \sum^{n}_{i=1}\1_{\boldw}(\boldx)_ii^m.$ Then, according to Lemma \ref{lemma:complexanalysisgeneral}, there exist complex numbers $z_1,\ldots,z_{12k}$, such that $|\frac{z_j-q}{1-q}|^n\le 2$ for $j\in[12k]$  and \eqref{equation:complexanalysisgeneral} holds for sufficiently large $n$.
 Lemma \ref{lemma:multibitgeneral} and Eq. \eqref{equation:complexanalysisgeneral} imply that
\begin{align*}
	\sum_{1\le i_1<\ldots<i_{12k}\le n}\Big|&\mathbb{E}_{\tilde{X}}\big[\1_{\tilde{X}_{i_j}=w_j,\forall j\in[12k]}\big]-\mathbb{E}_{\tilde{Y}}\big[\1_{\tilde{Y}_{i_j}=w_j,\forall j\in[12k]}\big]\Big|\\
	&(1-q)^{-12k} (\frac{z_1-q}{1-q})^{i_1}\prod^{12k}_{j=2}(\frac{z_j-q}{1-q})^{i_j-i_{j-1}-1}\\
	&\ge \frac{1}{n^{O(k)}},
\end{align*}
and thus that
\begin{align*}
	&\max_{1\le i_1<\ldots<i_{12k}\le n}\Big|\mathbb{E}_{\tilde{X}}\big[\1_{\tilde{X}_{i_j}=w_j,\forall j\in[12k]}\big]-\mathbb{E}_{\tilde{Y}}\big[\1_{\tilde{Y}_{i_j}=w_j,\forall j\in[12k]}\big]\Big|\\
	&\ge \frac{1}{n^{O(k)}}\cdot (1-q)^{12k}\cdot\frac{1}{\binom{n}{12k}}\cdot \prod^{12k}_{j=1}\min\big\{1,|\frac{1-q}{z_j-q}|^n\big\}\\
	&=\frac{1}{n^{O(k)}}.
\end{align*}
Therefore, \eqref{equation:totalvariance} holds and the proof is done.
\section{Proof of Lemma \ref{lemma:complexanalysis}}\label{section:proofoflemma6}
Without loss of generality, assume that $m$ is the smallest non-negative integer satisfying $\sum^{n}_{i=1}\1_{\boldw}(\boldx)_ii^m\ne \sum^{n}_{i=1}\1_{\boldw}(\boldx)_ii^m$. Let
\begin{align}
	f(z)=\sum^{n}_{i=1}\1_{\boldw}(\boldx)_iz^i
	-\sum^{n}_{i=1}\1_{\boldw}(\boldx)_iz^i
\end{align}
be a complex polynomial. The coefficients of $f(z)$ are within the set $\{-1,0,1\}$. 

According to Lemma \ref{proposition:divide}, we have that $f(z)=(z-1)^mq(x)$, where $q(z)=\sum^{n_1}_{i=0}c_iz^i$ is a complex polynomial with integer coefficients and $(z-1)$ does not divide $q(z)$, i.e., $q(1)\ne 0$.
The following result was presented in \cite{BE}. It gives an upper bound on the norm of coefficients of $q(z)$. 
\begin{proposition}\label{proposition:BE} \cite{BE}
	If a complex degree $n$ polynomial $f(z)$ has all coefficients with norm not greater than $1$, and can be factorized by 
	\begin{align*}
		f(z)=(z-1)^mq(z)=(z-1)^m(c_{n_1}z^{n_1}+\ldots+c_0),
	\end{align*}
then, we have that
	$\sum^{n_1}_{i=1}|c_{i}|\le (n+1)(\frac{en}{m})^{m}$.
\end{proposition}
We are now ready to prove Lemma \ref{lemma:complexanalysis}. Let $D\triangleq  2m+2$ and $z_j=\exp(\frac{2j\pi i}{D})$, $j\in[D]$ be a sequence of $D$ complex numbers equally distributed on a unit circle. We first show that there exists a number $j\in[D]$ satisfying
\begin{align*}
	q(1+\frac{z_j}{n^2})\ge \frac{1}{n^{O(m)}}.
\end{align*}
	Note that
		\begin{align*}
		\Big|\sum^D_{j=1}q(1+\frac{z_j}{n^2})\Big|=&\Big|\sum^{n_1}_{r=0}c_{r}\big[\sum^{D}_{j=1}(1+\frac{\exp(\frac{2j\pi i}{D})}{n^2})^r\big]\Big|\\
		=&\Big|\sum^{n_1}_{r=0}c_{r}\sum^{r}_{s=0}\binom{r}{s}\sum^{D}_{j=1}\frac{\exp(\frac{2js\pi i}{D})}{n^{2s}}\Big|\\
		\overset{(a)}{=}&\Big|\sum^{n_1}_{r=0}c_r\sum^{r}_{s=0}\binom{r}{s}\frac{D\1_{D\text{ divides }s}}{n^{2s}}\Big|\\
		=&\Big|\sum^{n_1}_{r=0}c_r(D+\sum^{r}_{s=1}\binom{r}{s}\frac{D\1_{D\text{ divides }s}}{n^{2s}})\Big|\\
		=&\Big|Dq(1)+\sum^{n_1}_{r=0}\sum^{r}_{s=1}c_r\binom{r}{s}\frac{D\1_{D\text{ divides }s}}{n^{2s}}\Big|\\
		\ge& D|q(1)|-\sum^{n_1}_{s=1}(\sum^{n_1}_{r=s}|c_r|)\binom{n}{s}\frac{D\1_{D\text{ divides }s}}{n^{2s}}\\
		\overset{(b)}{\ge} & D- (n+1)(\frac{en}{m})^{m}\sum^{n_1}_{s=1}\frac{D\1_{D\text{ divides }s}}{n^{s}}\\
		\ge&D -D(n+1)(\frac{en}{m})^{m}\frac{1}{n^{D}}\sum^{\infty}_{t=0}\frac{1}{n^{Dt}}\\
		\ge& D -D(n+1)(\frac{en}{m})^{m}\frac{2}{n^{D}}\\
		\overset{D\triangleq 2m+2}{=} & D-D(n+1)(\frac{e}{mn})^m\frac{2}{n^2} \\
		=&D-o(\frac{1}{n})\\
		\ge &1
	\end{align*}
for sufficiently large $n$, where (a) follows from the identity
\begin{align*}
\sum^{D}_{j=1}\exp(\frac{2js\pi i}{D})=	\frac{\exp(\frac{2sD\pi i}{D})-1}{\exp(\frac{2s\pi i}{D})-1}=D\1_{D\text{ divides }s}
\end{align*}
and 
(b) follows from Proposition \ref{proposition:BE} and the facts that $q(1)$ is a nonzero integer and that $\binom{n}{s}\le n^s$.
Therefore, there exists an integer $j$ such that
\begin{align*}
	|q(1+\frac{z_j}{n^2})|\ge \frac{1}{D},
\end{align*}
and thus that
\begin{align*}
	|f(1+\frac{z_j}{n^2})|=&\frac{|q(1+\frac{z_j}{n^2})|}{n^{2m}}\\
	\ge & \frac{1}{n^{2m}(2m+2)}.
\end{align*}  
Moreover, we have that 
\begin{align*}
	|\frac{1+\frac{z_j}{n^2}-q}{1-q}|^n=&\left(1+\frac{1}{n^4(1-q)^2}+2\frac{cos(\frac{2j\pi }{D})}{n^2(1-q )}\right)^\frac{n}{2}\\
	\le &2
\end{align*}
for sufficiently large $n$.
Hence, $z=1+\frac{z_j}{n^2}$ satisfies the conditions in Lemma \ref{lemma:complexanalysis}.
\section{Conclusion}\label{section:conclusion}
In this paper we studied the trace reconstruction problem when the string to be recovered is within bounded edit distance to a known string. Our result implies that when the edit distance is constant, the number of traces needed is polynomial. The problem of whether a polynomial number of samples suffices for the general trace reconstruction is open. However, it will be interesting to see if the methods in this paper can be extended to obtain more general results.
\bibliographystyle{IEEEtran}

\begin{thebibliography}{1}
\bibitem{Ban}
F.~Ban, X.~Chen, A.~Freilich, R.~A. Servedio, and S.~Sinha, ``Beyond trace reconstruction: Population recovery from the deletion channel.''  \emph{60th IEEE Annual Symposium on Foundations of Computer
	Science (FOCS)}, pp.~745–-768, 2019.

\bibitem{Batu}
T.~Batu, S.~Kannan, S.~Khanna, and A.~McGregor, ``Reconstructing strings from random traces.''  \emph{Proceedings of the Fifteenth Annual ACM-SIAM Symposium on Discrete Algorithms (SODA)}, pp.~910--918, 2004.

\bibitem{BPRS}  V.~Bhardwaj, P.~A. Pevzner, C.~Rashtchian, and Y.~Safonova, ''Trace reconstruction problems in
computational biology.'' \emph{IEEE Transactions on Information Theory}, to appear.

\bibitem{Borwein}
P.~Borwein, ``Computational excursions in analysis and number theory.''  \emph{Springer Science \& Business Media}, 2012.

\bibitem{BE97}
P.~Borwein and T.~Erd{\'e}lyi, ''Littlewood-type problems on subarcs of the unit circle.''
\emph{Indiana University mathematics journal}, vol.~46, no.~4, pp.~1323–-1346, 1997.

\bibitem{BE}
Peter Borwein, Tam{\'a}s Erd{\'e}lyi, and G{\'e}za K{\'o}s. ''Littlewood-type problems on [0, 1].''
\emph{Proceedings of the London Mathematical Society}, vol.~79, no.~1, pp.~22--46, 1999.


\bibitem{Chasel}
Z.~Chase, ``New lower bounds for trace reconstruction.''  \emph{arXiv:1905.03031}, 2020.

\bibitem{Chaseu}
Z.~Chase, ``New upper bounds for trace reconstruction.''  \emph{arXiv:2009.03296}, 2020.

\bibitem{Ryan}
M.~Cheraghchi, R.~Gabrys, O.~Milenkovic, and J.~Ribeiro, ``Coded trace reconstruction.'' \emph{IEEE Transactions on Information Theory}, vol.~66, no.~10, pp.~6084–-6103, 2020.

\bibitem{Rocco}
X.~Chen, A.~De, C.~H. Lee, R.~A. Servedio, and S.~Sinha, ``Polynomial-time trace reconstruction in the low deletion rate regime.''  \emph{arXiv:2012.02844}, 2020.

\bibitem{Chen}
X.~Chen, A.~De, C.~H. Lee, R.~A. Servedio, and S.~Sinha, ``Polynomial-time trace reconstruction in the smoothed complexity model.''  \emph{Proceedings of the Fifteenth Annual ACM-SIAM Symposium on Discrete Algorithms (SODA)}, pp.~54--73, 2021.

\bibitem{Racz}
S.~Davies, M.~Z. R{\'a}cz, C.~Rashtchian and B.~G. Schiffer, ``Approximate trace reconstruction.''  \emph{arXiv:2012.06713}, 2020.

\bibitem{De}
A.~De, R.~O’Donnell, and R.~A. Servedio, ``Optimal mean-based algorithms for trace reconstruction.'' \emph{The Annals of Applied Probability}, vol.~29, no.~2, pp.~851--874, 2019.

\bibitem{Sudan}
E.~Grigorescu, M.~Sudan, and M.~Zhu, ``Limitations of mean-based algorithms for trace reconstruction at small distance.''  \emph{arXiv:2011.13737}, 2020.

\bibitem{Holden}
N.~Holden, R.~Pemantle, and Y.~Peres, ``Subpolynomial trace reconstruction for random strings and arbitrary deletion probability.''  \emph{Proceedings of the 31st Conference On Learning Theory (COLT)}, pp.~1799--1840, 2018.

\bibitem{Mitzenmacher}
T.~Holenstein, M.~Mitzenmacher, R.~Panigrahy, and U.~Wieder, ``Trace reconstruction with constant deletion probability and related results.''  \emph{Proc. 19th ACM-SIAM Symposium on Discrete
	Algorithms (SODA)}, pp.~389–-398, 2008.

\bibitem{Mazumdar}
A.~Krishnamurthy, A.~Mazumdar, A.~McGregor, and S.~Pal, ``Trace reconstruction: Generalized and parameterized.''  \emph{arXiv:1904.09618}, 2019.

\bibitem{Lev}
V.~I. Levenshtein, ``Efficient reconstruction of sequences.'' \emph{IEEE Transactions on Information Theory}, vol.~47, no.~1, pp.~2--22, 2001.


\bibitem{KR97}
I.~Krasikov and Y.~Roditty, ``On a reconstruction problem for sequences.'' \emph{Journal of Combina-
	torial Theory, Series A}, vol.~77, no.~2, pp.~344--348, 1997.

\bibitem{Peres}
F.~Nazarov and Y.~Peres, ``Trace reconstruction with $\exp(O(n^{1/3}))$ samples.''  \emph{Proceedings of the
	49th Annual ACM SIGACT Symposium on Theory of Computing (STOC)}, pp.~1042--1046, 2017.




\bibitem{Microsoft}
L.~Organick, S.~D.~Ang, Y.~J.~Chen, R.~Lopez, S.~Yekhanin, K.~Makarychev, M.~Z.~Racz, G.~Kamath, P.~Gopalan, B.~Nguyen, C.~Takahashi, S.~Newman, H.~Y.~Parker, C.~Rashtchian, G.~G.~K.~Stewart, R.~Carlson, J.~Mulligan, D.~Carmean, G.~Seelig, L.~Ceze, and K.~Strauss, ``Scaling up DNA data storage and random access retrieval,'' \emph{bioRxiv}, 2017.


\bibitem{Zhai}
Y.~Peres and A.~Zhai, ``Average-case reconstruction for the deletion channel: Subpolynomially many traces suffice.''  \emph{58th IEEE Annual Symposium on Foundations of Computer
	Science (FOCS)}, pp.~228--239, 2017.

\bibitem{Robson}
J.~M. Robson, ``Separating strings with small automata.'' \emph{Information Processing Letters}, vol.~30, no.~4, pp.~209--214, 1989.

\bibitem{SB20}  J. Sima and J. Bruck, ``Optimal $k$-deletion correcting codes,'' \emph{IEEE Transactions on Information Theory}, to appear.

\bibitem{Wiki}
Wikipedia, "Reference genome", available at
\url{https://en.wikipedia.org/wiki/Reference_genome}

\bibitem{Yazdi}
S.~M.~H.~T.~Yazdi, R.~Gabrys, and O.~Milenkovic, ``Portable and error-free DNA-based data storage,'' \emph{Scientific reports}, vol.~7, no.~1, p.~5011, 2017.
\end{thebibliography}

\end{document}